\documentclass[conference]{IEEEtran}
\IEEEoverridecommandlockouts

\usepackage[pdftex]{graphicx}

\usepackage{epstopdf}

\usepackage{amsmath}

\usepackage{algorithm}

\usepackage[noend]{algpseudocode}

\usepackage{fancyhdr}

\fancypagestyle{firstpage}
{
    \fancyhf{}  
    
    \fancyhead[C]{To appear at the 27th Design, Automation and Test in Europe Conference (DATE'24), Mar 25-27 2024, Valencia, Spain.}
    \fancyfoot[C]{\footnotesize{\copyright 2023 IEEE. Personal use of this material is permitted. Permission from IEEE must be obtained for all other uses, in any current or future media, including reprinting/republishing this material for advertising or promotional purposes, creating new collective works, for resale or redistribution to servers or lists, or reuse of any copyrighted component of this work in other works.}}
}

\fancypagestyle{otherpages}
{
    \fancyhf{}
    
    \fancyfoot[C]{\footnotesize{\copyright 2023 IEEE. Personal use of this material is permitted. Permission from IEEE must be obtained for all other uses, in any current or future media, including reprinting/republishing this material for advertising or promotional purposes, creating new collective works, for resale or redistribution to servers or lists, or reuse of any copyrighted component of this work in other works.}}
}

\usepackage[dvipsnames]{xcolor}

\newcommand{\refsec}[1]{Sec.~\ref{sec:#1}}
\newcommand{\refalg}[1]{Alg.~\ref{alg:#1}}

\usepackage{etoolbox}
\usepackage[normalem]{ulem}

\usepackage{amsthm}
\newtheorem{theorem}{Theorem}
\theoremstyle{definition}

\usepackage[nolist]{acronym}

\usepackage[a4paper, total={184mm,239mm}]{geometry}
\def\BibTeX{{\rm B\kern-.05em{\sc i\kern-.025em b}\kern-.08em
    T\kern-.1667em\lower.7ex\hbox{E}\kern-.125emX}}

\begin{document}

\begin{acronym}
    \acro{hw}[HW]{hardware}
    \acro{sw}[SW]{software}
    \acro{ip}[IP]{Intellectual Property}
    \acrodefplural{ip}[IPs]{Intellectual Properties}
    \acro{soc}[SoC]{System on Chip}
    \acrodefplural{soc}[SoCs]{Systems on Chip}
    \acro{3pip}[3PIP]{third-party IP}
    \acro{ht}[HT]{hardware Trojan}
    \acro{ipc}[IPC]{Interval Property Checking}
    \acro{cex}[CEX]{counterexample}
    \acro{uci}[UCI]{Unused Circuit Identification}
    \acro{bmc}[BMC]{Bounded Model Checking}
    \acro{ift}[IFT]{Information Flow Tracking}
    \acro{ml}[ML]{Machine Learning}
  \end{acronym}

\title{
A Golden-Free Formal Method for Trojan Detection in Non-Interfering Accelerators%
\thanks{This work was supported partly by Bundesministerium f\"ur Bildung und Forschung Scale4Edge, grant 
no. 16ME0122K-16ME0140+16ME0465, by Intel Corp., Scalable Assurance Program and by Siemens EDA.}
}

\author{
    \IEEEauthorblockN{%
    Anna Lena Duque Ant\'on\IEEEauthorrefmark{1},
    Johannes M\"uller\IEEEauthorrefmark{1},
    Lucas Deutschmann\IEEEauthorrefmark{1},
    Mohammad Rahmani Fadiheh\IEEEauthorrefmark{2},\\
    Dominik Stoffel\IEEEauthorrefmark{1},
    Wolfgang Kunz\IEEEauthorrefmark{1} \\%
  }
  \IEEEauthorblockA{%
        \IEEEauthorrefmark{1}University of Kaiserslautern-Landau, Kaiserslautern, Germany
        \IEEEauthorrefmark{2}Stanford University, Stanford, USA \\
        \vspace{-3ex}
  }%
    \IEEEauthorblockA{%
        \{anna.duqueanton, johannes.mueller, lucas.deutschmann, dominik.stoffel, wolfgang.kunz\}@rptu.de, fadiheh@stanford.edu
        \vspace{-3ex}
  }%
}
\maketitle
\thispagestyle{firstpage}
\pagestyle{otherpages}

\begin{abstract}
    The threat of hardware Trojans (HTs) in security-critical IPs like cryptographic
    accelerators poses severe security risks. %
    The HT detection methods available today mostly rely on golden models and
    detailed circuit specifications. 
    Often they are specific to certain HT payload types, making pre-silicon verification difficult and
    leading to security gaps. %
    We propose a novel formal verification method for HT detection in
    non-interfering accelerators at the Register Transfer Level (RTL),
    employing standard formal property checking. %
    Our method guarantees the exhaustive detection of any sequential HT independently of its payload behavior, including physical side channels.
    It does not require a golden model or a functional
    specification of the design. 
    The experimental results demonstrate efficient and effective
    detection of all sequential HTs in accelerators available on
    Trust-Hub, including those with complex triggers and payloads. %
\end{abstract}


\section{Introduction}
\label{sec:introduction}
Many of today's \acp{soc} outsource complex computation tasks to
  specialized hardware accelerators. %
  Such blocks of \ac{ip} can implement certain functions with higher
  performance and efficiency than \ac{sw} solutions using a CPU. %
The resulting need for specialized and cost-effective accelerators is catered to by a global supply chain. %
Accelerator \acp{ip} can be acquired and integrated
as \acp{3pip} or generated using third-party EDA tools. %
  However, the flexibility and the opportunities of a diverse and global supply chain also introduce new security risks. %
  Among these, \acp{ht} are a prominent class of threats~\cite{2014-BhuniaHsiao.etal}. %
  The threat is exacerbated by the trend to outsource even security-critical computations to \acp{3pip}, including implementations for %
  encryption, which is the foundation of the overall system security. %
Encryption accelerators, therefore, need to undergo thorough verification for any malicious behavior. %
For reasons of cost, verification time and coverage, this process is increasingly moved to the pre-silicon design phase~\cite{2014-NishaMerli.etal}. %

So far, however, classical verification methods often perform poorly in detecting \acp{ht}~\cite{2016-XiaoForte.etal}. %
Intelligent adversaries construct \emph{stealthy \acp{ht}} that are able to evade common detection techniques. %
A stealthy \ac{ht} executes its malicious behavior, commonly referred
to as \textit{payload}, only after a \textit{trigger} condition is
met. %
    We distinguish %
  between combinational and sequential \acp{ht}. %
While there exist effective detection methods for the former type~\cite{2018-ZhouGuin.etal}, detecting the latter type is still a hard problem~\cite{2021-JainZhou.etal}. %
  For sequential \acp{ht} the %
  trigger condition is %
    created %
  such that it requires a potentially long sequence of input events which has a very low
  probability of occurring during testing. %
While contemporary detection methods like functional validation can be quite effective against HTs with short trigger sequences,
more complex triggers, especially those based on very long input sequences, can easily neutralize such methods~\cite{2020-XueGu.etal}. %
Furthermore, many detection methods rely on golden models, i.e., an
\ac{ht}-free design, or detailed
specifications~\cite{2014-RathmairSchupfer.etal}, which may not be
available. %
And finally, most methods are limited with respect to what payload
types they can detect. %
Physical side channels, in particular, are largely ignored by previous
formal %
pre-silicon %
HT detection methods. %

We intend to overcome these challenges by providing a formal
verification method for \ac{ht} detection in %
accelerator %
\acp{ip}. %
We target \emph{non-interfering} accelerators (cf.~Sec.~\ref{sec:background}). %
In fact, many loosely-coupled accelerators integrated in heterogeneous SoCs
belong to this class~\cite{2020-SinghLonsing.etal}. %
Our method operates on RTL models and is based on standard formal
property checking, which makes it easy to integrate it into existing
verification flows. %
The proposed approach can produce formal guarantees for the absence of
\acp{ht} in a design. %

In summary, this paper makes the following contributions: %
\begin{itemize}
\item We propose, %
for the first time, a formal methodology that allows us to \textit{exhaustively}
verify the absence of %
any sequential \ac{ht} with an arbitrary long and complex trigger sequence in a non-interfering accelerator IP. (\refsec{method}) %
\item Our method does not rely on a golden model or a functional specification of the design, by merit of the proposed 2-safety computational model. (\refsec{intuition})%
\item We guarantee the detection of any sequential \ac{ht}  independently of its payload. %
We exploit that sequential HTs have some RTL representation of their payload, even in the case of physical side channels. (\refsec{payloads})%
\item We effectively and soundly decompose the verification target into single-cycle properties allowing us to introduce a scalable iterative verification flow (\refsec{flow}). %
We demonstrate the efficiency and effectiveness of our method by
application to all accelerator \acp{ip} available on
Trust-Hub~\cite{2013-SalmaniTehranipoor.etal, 2017-ShakyaHe.etal} (\refsec{experiments}). %
\end{itemize}

\section{Related Work}
\label{sec:related-work}
Several works leverage verification tests for hardware trojan
detection. %
In~\cite{2010-HicksFinnicum.etal}, malicious circuit detection is based on \ac{uci}. %
\ac{uci} identifies circuit %
parts that do not affect the outputs during verification tests and
therefore may include malicious logic. %
However, as demonstrated by~\cite{2011-SturtonHicks.etal}, the
adversary can design \acp{ht} in such a way that there is a
verification test affecting the \ac{ht}'s logic without fulfilling the
trigger condition. %
VeriTrust~\cite{2013-ZhangYuan.etal} detects HTs in a design by
analyzing system states that are not covered by verification tests to identify the
trigger. %
In~\cite{2013-WaksmanSuozzo.etal}, circuit wires are analyzed regarding their probability to influence outputs. %
The lower their impact the more likely they belong to an \ac{ht}. %
However,~\cite{2013-ZhangYuan.etal} and~\cite{2013-WaksmanSuozzo.etal} do not detect \acp{ht} with more complex sequential trigger logic, as discussed in~\cite{2014-ZhangYuan.etal}. %
In contrast, this paper proposes an approach that guarantees the detection of \acp{ht} with arbitrary long and complex trigger sequences. %
    
Other works apply formal verification for \ac{ht} detection. %
The work of~\cite{2018-HuArdeshiricham.etal}, for example, leverages \ac{ift} to derive formal models that are used for checking security properties
like confidentiality or integrity of sensitive data. %
While the approaches are sound w.r.t. to their target properties, they
depend on modeling specific payloads and are, thus, not exhaustive. %
In particular, they cannot detect HTs based on non-functional design
specifications like (physical) side channels. %
The same limitations exist in approaches that use \ac{bmc} to detect
\acp{ht} that modify data in
registers~\cite{2015-RajendranVedula.etal} or leak secret data to the
IP outputs. %
Furthermore, due to the limitations of the bounded proof, the approaches are unable to detect trojans with very long trigger sequences. %
For the same reason, functional verification approaches based on \ac{bmc}~\cite{2020-SinghLonsing.etal} are not effective for trojan detection.
A SAT-based method is proposed in~\cite{2017-FernSan.etal} to detect
HTs that do not violate design specifications. %
The method detects \acp{ht} that modify unspecified design functionality.
Our approach is not limited to a specific \ac{ht} implementation, but detects arbitrary, possibly unknown, payload implementations.

The approaches in~\cite{2021-ItoUeno.etal, 2017-FarahmandiHuang.etal} use equivalence checking to compare the IP under verification against a golden, HT-free design, which, however, is often unavailable. %
Therefore, we propose a golden-free detection method.
  
In recent years, pre-silicon \ac{ht} detection based on \ac{ml} has
become popular. %
These approaches employ \ac{ml} to generate test  patterns that are more
likely to trigger \acp{ht}~\cite{2021-PanMishra} or to classify
structural design features as HT-free~\cite{2020-YuGu.etal,
2022-HeppBaehr.etal}. %
While these approaches are effective even for large designs, they are
inherently not exhaustive and thus cannot provide security
guarantees. %

\section{Trojans in Non-Interfering Accelerators}
\label{sec:background}

Our method focuses on the detection of HTs in so-called \emph{non-interfering} accelerator IPs. %
This notion, introduced by~\cite{2020-SinghLonsing.etal}, refers to the typical characteristic of accelerators that the computed result for a given input is independent of any inputs received earlier or later.  
It is important to note that this definition is not a restriction to combinational circuits but includes sequential ones. %
Fig.~\ref{fig:trojan-crypto} illustrates such a design. 
It shows a cryptographic accelerator infested with an \ac{ht}. %
The IP implements encryption of a plaintext with a key. %

\begin{figure}
    \centering
    \includegraphics[trim={72mm 62mm 72mm 62mm}, clip, width=\linewidth]{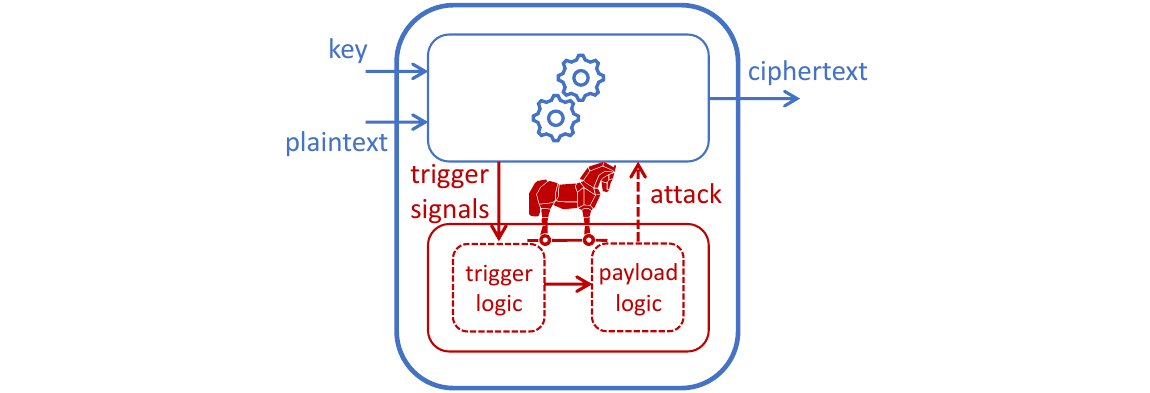} 
    \caption{HT in a cryptographic IP. %
The HT consists of a trigger and payload.} %
\vspace{-5mm}
    \label{fig:trojan-crypto}
\end{figure}
  
\acp{ht} can be classified by their trigger and their payload. %
The purpose of the trigger is to provide a reliable but hard-to-detect means for an
attacker to activate the trojan. %
To this end, an \ac{ht} trigger may rely on multiple logic signals of the
circuit and can consist of arbitrary complex (sequential) logic. %
When an \ac{ht} is activated, it executes its payload, which manifests
itself as malicious behavior in various ways, such as leaking a key
via output pins or using a power side channel. %

It is common to distinguish between \emph{combinational} and \emph{sequential} \acp{ht}.
There exist works that effectively detect combinational \acp{ht}, i.e., \acp{ht} that are triggered by a combinational circuit~\cite{2021-JainZhou.etal}.
However, the detection of sequential \acp{ht}, i.e., \acp{ht} that are triggered by a sequential circuit is still an open challenge~\cite{2011-WangNarasimhan.etal}.
Sequential \acp{ht} are activated only after a possibly long
sequence of regular, specification-conforming executions. %
The trigger does not necessarily depend on the input values. %
It may also simply be a counter activated by the reset signal. %

\section{Method}
\label{sec:method}

In the following we present a formal property checking method that allows us to detect
\textit{all} sequential \acp{ht} in accelerator
\acp{ip}. %
The properties are design-agnostic and do not require a golden model of the design. %
We describe the intuition behind our analysis and the key challenges
in Sec.~\ref{sec:intuition}. %
In Sec.~\ref{sec:trigger} and~\ref{sec:payloads}, we then describe how
the proposed methodology meets these challenges. %
In addition, we discuss the exhaustiveness of our method in
Sec.~\ref{sec:exhaustiveness}. %

\subsection{Intuition}
\label{sec:intuition}
As discussed in Sec.~\ref{sec:background}, an \ac{ht} executes its payload after being activated by its
trigger. %
The payload can be any malicious behavior that modifies the
functionality of the \ac{ip} or adds unwanted functionality to the
\ac{ip}. %
Instead of comparing the \ac{ip}'s behavior with a golden model, which is usually unavailable, our
method compares \textit{two identical instances} of the \ac{ip}: %
We analyze the behavior of the two instances under the same inputs,
but allow the solver to assume %
different, arbitrary \textit{input histories}. 
These different input histories are captured by the different \textit{symbolic starting states} in the two instances of our computational model, as elaborated below. %
In case the design is infested with a sequential~\ac{ht}, we can compare one
instance with a triggered HT about to execute its payload to one with
an untriggered \ac{ht}. %
The method then verifies whether the two instances behave the same for
all possible input sequences from that time point on. %
If the proof fails, it returns a counterexample pointing at the HT's
payload. %
If the verification does not detect any divergence in the state or
output of the two IP instances, then there is no \ac{ht} with a
malicious payload in the design. %
This is illustrated with the 2-safety \textit{miter} setup
shown in Fig.~\ref{fig:trojan-detection}. %

\begin{figure}
  \centering
  \includegraphics[trim={70mm 60mm 70mm 60mm}, clip, width=\linewidth]{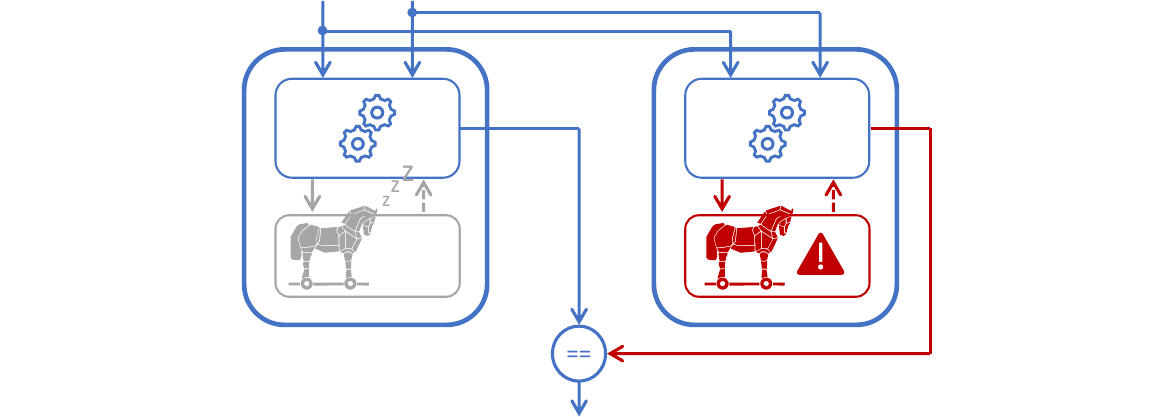} 
  \caption{
  Two instances of the same IP (``miter''), both containing an HT. %
  The same inputs are fed to both instances. %
  The HT in instance\textsubscript{1} is triggered while the HT in instance\textsubscript{2} is dormant. 
  This results in an observable difference in the behavior of the two IPs.%
  }
  \label{fig:trojan-detection}
  \vspace{-5mm}
\end{figure}

There are two key challenges with this approach that our method addresses: %
\begin{enumerate}
\item modeling arbitrary (unknown) trigger sequences, including ones
  of arbitrary length, and %
  allowing the \ac{ht} to be triggered in
  one of the instances while not in the other; %
  
\item targeting an arbitrary payload. %
  
\end{enumerate}
In Secs.~\ref{sec:trigger} and~\ref{sec:payloads} respectively, we elaborate
on how these two challenges are overcome. %

\subsection{Modelling trigger sequences}
\label{sec:trigger}

The key idea for dealing with \ac{ht} triggers exploits the nature of our employed property checking technique. %
Our detection method is based on \textit{\ac{ipc}}~\cite{2008-NguyenThalmaier.etal}, which uses bounded properties
but achieves unbounded proof validity. %
\ac{ipc} facilitates proving (interval) properties of the form \textit{antecedent} $\Rightarrow$ \textit{consequent} with a symbolic starting state. 
This means that a solver verifying an interval property is free to choose any state of the design as
starting state for the proof as long as the antecedent is fulfilled. %
This symbolic starting state can thus \textit{implicitly} model any
history of previous states, including \emph{any possible trigger
  sequence}. %
Essentially, our method leverages \ac{ipc} to ``fast-forward'' to the time point at which a potential \ac{ht} is activated. %

  The verification method presented in the following sections exploits an important characteristic of non-interfering accelerators. %
    Such designs are often \emph{data-driven}, i.e., they typically %
  determine the internal states relevant for their computations only from the inputs and, essentially, independently %
  of the
  accelerator's internal state at the start of the computation. %
  This allows for the effective use of symbolic starting states in the
  properties of our verification flow without the risk of many false
  alarms. %

\subsection{Detecting payloads}
\label{sec:payloads}

Using a symbolic starting state, we can model an active \ac{ht} in one
instance together with a dormant \ac{ht} in the other instance. %
This addresses the first challenge from \refsec{intuition}. %
However, the second challenge -- the variety of \ac{ht} payloads -- still remains. %
  We aim to detect not just specific payloads, such as leaking a secret key, but %
any possible sequential \ac{ht} visible at the RTL, without requiring a full specification. %

  Consider a basic IPC proof that enumerates all malicious HT
  behaviors without reference to a golden model. %
  Some functionality
  might be overlooked, and HT payloads with additional behaviors or
  physical side channel manifestations could be missed. %
We overcome these problems by leveraging an important observation: %
Regardless of its nature, the effect of a sequential HT's payload has
to manifest itself in at least one state signal or output of the
design or it does not have any security implications. %
Instead of formulating all possible behaviors in the consequent, we
use this observation together with our employed miter setup. %

We verify the equivalence of the two instances using a structural
decomposition of the design. %
We compute all state and output signals that appear in the transitive
fanout cone of the inputs and partition them according to the smallest
number of clock cycles it takes for the inputs to affect their
value.
  We denote the set of \textit{state} and
  \textit{output} signals affected after $k$ clock
  cycles by \textit{fanouts\_CC\textsubscript{k}}.
We then prove for each $k$ the equivalence of \textit{fanouts\_CC\textsubscript{k}} between
the two instances for the corresponding time window. %
This results in the \textit{trojan\_property} in Fig.~\ref{fig:refined-trojan-property}. %
The property checker is free to choose arbitrary starting states, as
long as, at time point~$t$, the same inputs are provided to the two IP
instances. %
For each consecutive clock cycle, the equality of the corresponding \textit{fanouts\_CC\textsubscript{k}} is checked. %
Proving this interval property makes any malicious behavior visible that manifests itself in either state or output signals in the fanout path of the inputs. %
We provide a detailed pseudo-code description for our formal
verification flow in~\refsec{flow}. %

\begin{figure}%
      \begin{minipage}{0.9\linewidth}
       \fontsize{9}{11}\selectfont
         \begin{tabbing}
          XX\=XXXXX\=XXX\=XXX\=\kill
          \textbf{trojan\_property:}\\
          \textcolor{blue}{assume:} \\
           \> \textcolor{blue}{at} $t$: %
          \>\>$\text{\textit{inputs\_instance\textsubscript{1} = inputs\_instance\textsubscript{2}}}$ \\
          \textcolor{blue}{prove:} \\
          \> \textcolor{blue}{at} $t+1$: \>\>$\text{\textit{fanouts\_CC\textsubscript{1}\_instance\textsubscript{1} = fanouts\_CC\textsubscript{1}\_instance\textsubscript{2}}}$ \\
          \> \textcolor{blue}{at} $t+2$: \>\>$\text{\textit{fanouts\_CC\textsubscript{2}\_instance\textsubscript{1} = fanouts\_CC\textsubscript{2}\_instance\textsubscript{2}}}$ \\  
          \> \ldots \\
          \> \textcolor{blue}{at} $t+n$: \>\>$\text{\textit{fanouts\_CC\textsubscript{n}\_instance\textsubscript{1} = fanouts\_CC\textsubscript{n}\_instance\textsubscript{2}}}$
        \end{tabbing}
      \end{minipage}
  \caption{Interval property for hardware trojan detection}
  \label{fig:refined-trojan-property}
  \vspace{-5mm}
  \end{figure}  

\subsection{Exhaustiveness}
\label{sec:exhaustiveness}

The question arises whether we can exhaustively detect \emph{any} malicious
behavior introduced by a sequential HT with this property. %
For this, we need to distinguish three cases: %
\begin{enumerate}
\item The payload affects at least one state or output signal that lies on a fanout path of the inputs. %
\item The payload affects only state or output signals that are not
  part of a fanout path. %
  \item The payload has no effect on any state or output signal. %
\end{enumerate}
\textit{Case 1:} %
The state or output signal, say~$s$, is covered by the consequent
(\textit{prove} part) of our property. %
  Hence, a symbolic initial state exists where
  instance\textsubscript{1} meets the HT trigger condition while
  instance\textsubscript{2} does not. %
  As a result, the triggered HT affects~$s$ in
  instance\textsubscript{1} but does not in
  instance\textsubscript{2}. %
  The property checker computes a counterexample (CEX) that
  demonstrates this difference, unveiling the payload's malicious
  behavior.%

\textit{Case 2:} %
The \ac{ht} is activated independently of the inputs%
, e.g., a timer started by the system reset and its payload does not
affect the fanout cone of the IP's inputs (cf.~\refsec{experiments}
and design example AES-T1900). %
This case is not detectable by the property. %
However, it can be covered by a simple structural analysis: %
We need to check whether all state and output signals of our \ac{ip} appear in
the \textit{prove} part of our property. %
Those that do not may belong to a possible \ac{ht}. %

\textit{Case 3:} If the \ac{ht}'s payload neither manifests itself in any state signal nor in any output signal then the \ac{ht} does not implement security-critical behavior visible at the RTL. %

Since cases 1 to 3 cover all possible cases of an \ac{ht}'s payload in the RTL design we conclude that we exhaustively detect all sequential \acp{ht} with
this method. %

\section{Formal Detection Flow}
\label{sec:flow}
For practical purposes, we seek to develop an \ac{ht} analysis that is scalable and easy to use for the verification engineer. %
We decompose the property in Fig.~\ref{fig:refined-trojan-property} into a set of interval properties each covering exactly one clock cycle. %
These become elements of an iterative verification flow where
individual proofs have short runtimes and counterexamples point to
potentially malicious behavior with high accuracy. %

We employ two types of properties. %
The first type is the \textit{init\_property} shown in
Fig.~\ref{fig:init-property}. %
It verifies that there is no malicious interference with the propagation
of the input signals to the first %
fanout signals, \textit{fanouts\_CC\textsubscript{1}}, %
  reached in the IP. %
The property assumes equal inputs, under which the equality of the
\textit{fanouts\_CC\textsubscript{1}} %
  is
proven. %
The \textit{init\_property} is a cutout of the
\textit{trojan\_property} of Fig.~\ref{fig:trojan-detection} until
time point~$t+1$. %
The second type of property, the %
fanout properties
  (%
\textit{fanout\_property\_k}%
)%
, %
  shown %
in~Fig.~\ref{fig:fanout-property}%
  , %
covers the equality checks for
all subsequent time points. %
As we %
  show %
in~Sec.~\ref{sec:soundness}, the set of decomposed
properties is equivalent to the \textit{trojan\_property}. %

Alg.~\ref{alg:HT-detection-flow} %
  implements %
the iterative \ac{ht} detection flow based on these two property types
and the method described in Sec.~\ref{sec:method}. %
In the first step, the %
\textit{fanouts\_CC\textsubscript{1}}%
, i.e., all state and output %
signals reachable within one clock cycle from the input signals, are
computed: %
Get\_Fanout() implements a simple structural analysis that traces
syntactic dependencies of %
  state-holding %
elements in the RTL design. %
With this information, the \textit{init\_property} is constructed %
  (line~3)%
. %
The property is then checked with IPC. %
In case the property fails, a counterexample, \textit{CEX}, is
returned %
    that %
  points to a possible hardware trojan and %
    that %
  must be inspected by the verification engineer. %
  If the property holds, the procedure continues. %
  In the next step the %
  \textit{fanouts\_CC\textsubscript{1}} become the starting points of the structural analysis. %
  All state %
   and output %
  signals reachable within one clock cycle from the %
  \textit{fanouts\_CC\textsubscript{1}} %
  are determined and the corresponding \textit{fanout\_property\_1} is
  constructed for the next iteration. %
  Any failing interval property will produce a \ac{cex} that shows the
  exact state signals where a potential trojan might be implemented. %
  This process is repeated, verifying %
  a \textit{fanout\_property\_k}
  in each iteration, until no new %
  state or output signals %
  are added. %
In case all properties hold, we conclude the procedure by checking whether the property set covers all state and output %
  signals of the IP under verification (cf.~case 2 of Sec.~\ref{sec:exhaustiveness}). %
If there are any state or output %
  signals left, the set of uncovered %
signals~(\textit{UCS}) is returned. %
It is important to note that the number of loop iterations (line 8-16) is limited by the structural, not the sequential, depth of the design.

  \begin{algorithm}
    \caption{\small Formal HT Detection Flow}\label{alg:HT-detection-flow}
     \small
    \begin{algorithmic}[1]
        \Procedure{HT-Detection}{\textit{IP, inputs}}
        \State \textit{fanouts\_CC\textsubscript{1}} $\gets$ Get\_Fanout(\textit{IP, inputs})
        \State \textit{init\_property} $\gets$ Create\_Init\_Property(\textit{inputs, fanouts\_CC\textsubscript{1}})
        \State \textit{CEX} $\gets$ IPC(\textit{init\_property})
        \If {$\text{\textit{CEX}} \neq \emptyset$}
        \Return~\textit{CEX} 
        \EndIf
        \State \textit{fanouts\_all} $\gets$ $\emptyset$
        \State $k$ $\gets$ $1$
        \Repeat
            \State \textit{fanouts\_all} $\gets$ \textit{fanouts\_all} $\cup$ \textit{fanouts\_CC\textsubscript{k}}
            \State \textit{fanouts\_CC\textsubscript{k+1}} $\gets$ Get\_Fanout(\textit{IP, fanouts\_CC\textsubscript{k}})
            \State \textit{fanout\_property\_k} $\gets$  Create\_Property(\textit{fanouts\_CC\textsubscript{k}}, 
            \begin{tabbing}
              XXXXXX\=XXXXX\=XXXXX\=XXXXX\=XXXXX\=\kill
            \>\>\>\>\> \textit{fanouts\_CC\textsubscript{k+1}})
            \end{tabbing}
            \State \textit{CEX} $\gets$ IPC(\textit{fanout\_property\_k})
            \If {$\text{\textit{CEX}} \neq \emptyset$}
            \Return~\textit{CEX}
            \EndIf
            \State \textit{fanouts\_CC\textsubscript{k}} $\gets$ \textit{fanouts\_CC\textsubscript{k+1}}
            \State $k$ $\gets$ $k+1$
        \Until{\textit{fanouts\_all} $\cup$ \textit{fanouts\_CC\textsubscript{k}} == \textit{fanouts\_all}}
        \State \textit{UCS} $\gets$ Check\_Signal\_Coverage(\textit{IP, fanouts\_all})
        \If {$\text{\textit{UCS}} \neq \emptyset$}
        \Return~$\text{\textit{UCS}}$  
        \EndIf
      \State \Return "SECURE"      
      \EndProcedure
    \end{algorithmic}
  \end{algorithm}
\begin{figure}[H]
\vspace{-5mm}
    \begin{minipage}{0.9\linewidth}
     \fontsize{9}{11}\selectfont
       \begin{tabbing}
        XX\=XXXXXX\=XXX\=XXX\=\kill
        \textbf{init\_property:}\\
        \textcolor{blue}{assume:} \\
        \> \textcolor{blue}{at} $t$:
        \>$\text{\textit{inputs\_instance\textsubscript{1} = inputs\_instance\textsubscript{2}}}$; \\
        \textcolor{blue}{prove:} \\
        \> \textcolor{blue}{at} $t+1$: \>$\text{\textit{fanouts\_CC\textsubscript{1}\_instance\textsubscript{1} = fanouts\_CC\textsubscript{1}\_instance\textsubscript{2}}}$; 
      \end{tabbing}
    \end{minipage}
\caption{Interval property for Init Check}
\label{fig:init-property}
\vspace{-3mm}
\end{figure} 
\begin{figure}[H]
    \vspace{-3mm}
        \begin{minipage}{0.9\linewidth}
         \fontsize{9}{11}\selectfont
           \begin{tabbing}
            XX\=XXXXXX\=XXX\=XXX\=\kill
            \textbf{fanout\_property\_k:}\\
            \textcolor{blue}{assume:} \\
            \> \textcolor{blue}{at} $t$:
            \>$\text{\textit{fanouts\_CC\textsubscript{k}\_instance\textsubscript{1} = fanouts\_CC\textsubscript{k}\_instance\textsubscript{2}}}$; \\
            \textcolor{blue}{prove:} \\
            \> \textcolor{blue}{at} $t+1$: \>$\text{\textit{fanouts\_CC\textsubscript{k+1}\_instance\textsubscript{1}
            = fanouts\_CC\textsubscript{k+1}\_instance\textsubscript{2}}}$; 
          \end{tabbing}
        \end{minipage}
    \caption{Interval property for Fanout Check}
    \label{fig:fanout-property}
    \vspace{-2mm} 
    \end{figure} 

\subsection{Soundness of Property Decomposition}
\label{sec:soundness}

In~Sec.~\ref{sec:exhaustiveness} we have discussed the exhaustiveness of the \textit{trojan\_property}. %
Verifying this property for a design guarantees detection of any sequential \ac{ht} that it might be infected with. %
In the following, we prove that the same guarantee is given by verifying the \textit{init\_property} and all computed fanout properties (\textit{fanout\_property\_k}). %

\begin{theorem}
  \label{theo:soundness}
  At least one of the fanout properties (\textit{fanout\_property\_k}) or the \textit{init\_property} fails (1)~iff the \textit{trojan\_property} fails (2). %
  \qed
\end{theorem}
\begin{proof}
  A key observation to keep in mind for the proof is that each set of
  fanout signals \textit{fanouts\_CC\textsubscript{k}} considered at time point $k$ %
\emph{is identical} in both property formulations: %
the aggregate \textit{trojan\_property} of Fig.~\ref{fig:refined-trojan-property} and
the decomposition into \textit{init\_property}
(Fig.~\ref{fig:init-property}) and %
the fanout properties %
(Fig.~\ref{fig:fanout-property}). %
  
We decompose the theorem into two implications and prove them
individually. %

\textit{(2)} $\Rightarrow$ \textit{(1)}: %
Assume the \textit{trojan\_property} fails because there is a state or output %
signal~$z$ with different values in the two instances at clock cycle
$t+k+1$ with $0 < k < n$. %
We further assume, w.l.o.g., that the proof commitments of all
preceding clock cycles are proven to hold, i.e., in particular,
  \textit{fanouts\_CC\textsubscript{k}} %
are equal between the two instances for the considered $k$. %
Hence, there must exist a state signal~$x$ in the fanin of~$z$ %
with $x \not\in$~\textit{fanouts\_CC\textsubscript{k}} which holds different values in the two instances. %
  (Remember that~$x$ cannot be a primary input because all fanout
  signals of primary inputs are covered by the \textit{init\_property} and
  commitment~$t+1$ of the \textit{trojan\_property}.) %
  Now consider the corresponding \textit{fanout\_property\_k}. %
  Since it considers the same set of signals \textit{fanouts\_CC\textsubscript{k}} %
  in its assumption
  as the aggregate property, $x$ is missing in this set as well,
  causing the property to fail. %

\textit{(1)} $\Rightarrow$ \textit{(2)}: %
We prove this implication by contradiction. %
Assume, w.l.o.g., \textit{fanout\_property\_k} fails because of inequality of
state signal~$z$, but the \textit{trojan\_property} holds. %
  This means that~$z$ must depend on at least one state signal~$x$ where $x \not\in$~\textit{fanouts\_CC\textsubscript{k}}. %
Furthermore, since the \textit{trojan\_property} holds, $z$ must be proven
  equal at clock cycle~$t+k+1$. %
  But this requires that the \textit{trojan\_property} must also prove the
  equality of~$x$ at the previous clock cycle~$t+k$, because the equality of~$z$ depends on the equality of~$x$. %
  However, if this is true, then~$x$ must be an element of \textit{fanouts\_CC\textsubscript{k}}. %
This contradicts the assumption from above and proves the claim. %
\end{proof}

\subsection{Analyzing Counterexamples}
\label{sec:counterexamples}

Although this occurs rarely, as explained in~Sec.~\ref{sec:trigger}, for some IPs the property checker may produce false alarms, i.e., the \textit{init\_property} or  a \textit{fanout\_property\_k} fails for some signal~$z$ although there is no HT in the design. %
For understanding  such cases, assume \textit{fanout\_property\_k} fails for~$z$. %
This means that~$z$ is affected by some other signal~$x$ in the fanin of~$z$, but~$x \not\in$~\textit{fanouts\_CC\textsubscript{k}}. %
In other words, $x$ is not proven to be equal between the two instances by the predecessor \textit{fanout\_property\_k-1}. %
This may happen in two scenarios: %
  \\
  (1)~We do not necessarily prove the fanout properties in topological order. %
  Therefore, \textit{fanout\_property\_k} may fail even though~$x$ is proven to be equal in another fanout property or the \textit{init\_property}. %
This scenario can be solved by changing the proof order of the fanout properties and adding equality for~$x$ to the assumptions of
  \textit{fanout\_property\_k}. %
   We omitted this procedure in \refalg{HT-detection-flow} to keep the presentation simple. %
  \newline (2) $x$ depends on values of previous computations but is
  not part of an \ac{ht}. %
  In this scenario the verification engineer receives a CEX pinpointing the exact
  behavior that demonstrates the dependency between~$x$ and~$z$. %
  This greatly helps in disqualifying the behavior as an \ac{ht}. %
  Similarly to the first scenario, equality for~$x$ can then be assumed in \textit{fanout\_property\_k}. %
  
  Fortunately, the characteristics of non-interfering accelerators allow for effective use of a symbolic initial state in our computational model.
  As we demonstrate in our experimental results, we encountered false counterexamples only in few cases which were easy to diagnose. %

  \section{Experiments}
\label{sec:experiments}

We applied our method to the accelerator IPs available on Trust-Hub~\cite{2017-ShakyaHe.etal}. 
We evaluated all accelerators, except for three with simple, combinational \acp{ht} which are not considered in this work. %
All are non-interfering. 
As can be seen in Tab.~\ref{table:results}, they implement different crypto algorithms and HTs with a variety of payloads and triggers. %
For the payloads, secret data is leaked via output pins (OUT), different implementations of power side channels (\textit{PSC}), through leakage currents (\textit{LC}) or by modulating the output signal on an unused pin creating a radio frequency (\textit{RF}) signal. %
In four cases, the payload consists of denial-of-service attacks (\textit{DoS}) that aim at rapidly draining the battery. %
Other payloads interfere with the encryption via bit flips of the ciphertext output. %
The triggers in the experiments depend either on a predefined plaintext sequence or on implementations of counters that count certain events. %

\begin{table}[t]
    \fontsize{8}{10}\selectfont
    \centering
    \caption{}%
    \begin{tabular}{p{1.9cm}p{0.9cm}p{1.8cm}p{1.9cm}}

      \hline\hline\rule{0pt}{1.1em}%
      Benchmark & Payload & Trigger & Detected by %
      \\
        \hline\rule{0pt}{1.1em}%
        AES-T100 & PSC & plaintext seq. & init\_property \\
        AES-T1000 & PSC & plaintext seq. & init\_property \\
        AES-T1100 & PSC & plaintext seq. & init\_property \\
        AES-T1200 & PSC & \# encryptions & init\_property \\
        AES-T1300 & PSC & plaintext seq. & init\_property \\
        AES-T1400 & PSC & plaintext seq. & init\_property \\
        AES-T1500 & PSC & \# encryptions & init\_property \\
        AES-T1600 & RF & plaintext seq. & init\_property \\
        AES-T1700 & RF & \# encryptions & init\_property \\
        AES-T1800 & DoS & plaintext seq. & init\_property \\
        AES-T1900 & DoS & \# encryptions & coverage check \\
        AES-T2000 & LC & plaintext seq. & init\_property \\
        AES-T2100 & LC & \# encryptions & init\_property \\
        AES-T2500 & bit flip & \# clock cycles & fanout\_property\_21 \\
        AES-T2600 & bit flip & \# values & fanout\_property\_7 \\
        AES-T2700 & bit flip & \# clock cycles & fanout\_property\_21 \\
        AES-T2800 & bit flip & \# values & fanout\_property\_11 \\
        AES-T200 & PSC & plaintext seq. & init\_property \\
        AES-T300 & PSC & plaintext seq. & init\_property \\
        AES-T400 & RF & plaintext seq. & init\_property \\
        AES-T500 & DoS & plaintext seq. & init\_property\\
        AES-T600 & LC & plaintext seq. & init\_property \\
        AES-T700 & PSC & plaintext seq. & init\_property \\
        AES-T800 & PSC & plaintext seq. & init\_property \\
        AES-T900 & PSC & \# encryptions & init\_property \\
        BasicRSA-T200 & DoS & plaintext seq. & init\_property \\
        BasicRSA-T300 & OUT & \# encryptions & init\_property \\
        BasicRSA-T400 & OUT & \# encryptions & init\_property \\
      \hline\hline
    \end{tabular}
    \label{table:results}
    \vspace{-4mm}
\end{table}

We successfully detected all HTs. %
Our method, as discussed in Sec.~\ref{sec:method}, is independent of
the specific characteristics of the HT's implementation. %
The HTs were detected by a failed init or fanout property or through
the coverage check. %
Each AES benchmark also includes an HT-free version. %
We successfully applied our method to the HT-free designs and verified them %
  to be %
secure with respect to sequential HTs. %
They required the proof of the init and fanout properties. %
We did not encounter any spurious counterexamples. %
For the RSA designs no HT-free version was available on Trust-Hub.
  We manually removed the \acp{ht} %
    from the designs %
  and afterwards %
    could verify the absence of any HTs%
  .
During the proof we encountered 2 spurious CEXs that we handled according to the proposed methods in~\refsec{counterexamples}.
The proof runtime of each property was within 1 to~3 seconds and the
memory usage was less than 1\,GB. %
All experiments were conducted on an Intel i7-8700 @
3.2\,GHz %
with
64\,GB %
RAM running Linux and the commercial property checker OneSpin~360~DV
by Siemens EDA. %

In the following, we explain our experiments in more depth using two of the benchmarks as examples. %

\begin{figure}%
    \centering
    \begin{minipage}[h]{.49\linewidth}
    \centering
     \includegraphics[trim={112mm 65mm 112mm 65mm}, clip, width=1\linewidth]{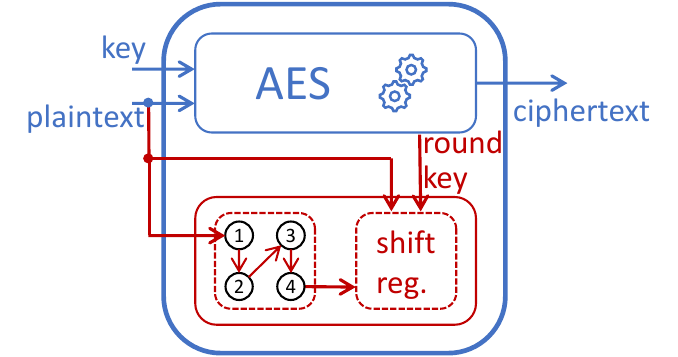} 
     \caption{AES-T1400}
     \label{fig:aes-t1400}
    \end{minipage}%
    \begin{minipage}[h]{.49\linewidth}
        \centering
        \includegraphics[trim={112mm 65mm 112mm 65mm}, clip, width=1\linewidth]{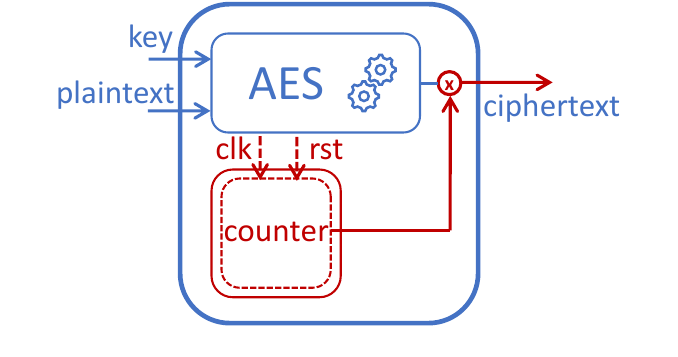} 
        \caption{AES-T2500}
        \label{fig:aes-t2500}
    \end{minipage}
    \vspace{-5mm}
\end{figure}

\textit{Example 1: %
AES-T1400.} 
The trojan in this example features a 4-state FSM as its trigger and is illustrated in Fig.~\ref{fig:aes-t1400}. %
The trojan is triggered when four specific plaintext inputs are observed in a specific order. %
Once activated, the trojan leaks, for each round of encryption, certain bits of the round key via a power side channel. %
The key bits are combined with known input bits and shifted into a register, thereby increasing power consumption. %
We detected the HT with a failed init property. %
The CEX provided by the property checker shows different values in the shift registers of the two instances. %

\textit{Example 2: %
AES-T2500.}
In the second example, the trojan is triggered by the fourth bit of a 4-bit synchronous counter. %
The counter itself does not depend on the IP inputs but starts counting from reset. %
After activation, the trojan flips the least significant bit (LSB) of the ciphertext output. %
Fig.~\ref{fig:aes-t2500} illustrates this behavior. %
The \ac{ht} is detected with \textit{fanout\_property\_21} which proves equal ciphertext outputs.
The property fails and the CEX shows the difference in the LSB of the ciphertext outputs due to a triggered \ac{ht} in only one of the two instances.

Even though it is not in the focus of this work, we conclude our experiments with demonstrating the potential of our method for HW IPs with more complex control behavior.
As an additional case study, we successfully applied the method also to a UART \textit{(RS232-T2400)} from the same benchmark suite. We detected the \ac{ht} by a failed fanout property. %
During the proof we encountered 3 spurious \acp{cex} that we could resolve by property re-verification (cf.~Sec.~\ref{sec:counterexamples}~(1)), and by
disqualifying the behaviors as non-malicious (cf.~Sec.~\ref{sec:counterexamples}~(2)).%
 
\section{Conclusion}
\label{sec:conclusion}

We introduce an exhaustive formal detection methodology for sequential \acp{ht} with arbitrarily long and complex trigger sequences. 
The proposed method is equally effective for any payload ranging from direct data leakage to power side channels.
It is a golden-free approach by merit of a 2-safety verification model.
The method is based on property checking and easy to integrate into pre-silicon verification.
The application of the method to a representative set of accelerators demonstrates the efficiency and effectiveness of our proof method. 
Future work will explore the class of interfering IPs such as (special-purpose) processors.

\ifCLASSOPTIONcaptionsoff
  \newpage
\fi

\bibliographystyle{IEEEtran}
{\footnotesize \bibliography{refs3}}

\end{document}